\documentclass[a4paper,USenglish]{article}

\usepackage{microtype}
\usepackage{xargs}
\usepackage{xspace}

\usepackage[utf8]{inputenc}
\usepackage{graphicx}
\usepackage{enumerate}
\usepackage{amsthm,amsmath,amsfonts}
\usepackage{url,hyperref}

\theoremstyle{plain}
\newtheorem{theorem}{Theorem}
\newtheorem{lemma}[theorem]{Lemma}

\newtheorem{proposition}[theorem]{Proposition}
\newtheorem{observation}[theorem]{Observation}

\theoremstyle{definition}
\newtheorem{definition}[theorem]{Definition}
\newtheorem{algo}[theorem]{Algorithm}

\newcommand{\bigO}{\mathcal{O}}

\newcommand{\MSOone}{\text{$\textbf{MSO}_1$}}
\newcommand{\MSOtwo}{\text{$\textbf{MSO}_2$}}

\DeclareMathOperator{\operatorClassNP}{NP}
\newcommand{\classNP}{\ensuremath{\operatorClassNP}}
\newcommand{\cwd}{\textsc{cwd}}

\newcommand{\PCG}{\textsc{Partial Complement  to  $\mathcal{G}$}\xspace}
\newcommand{\PCGshort}{\textsc{PC$\mathcal{G}$}\xspace}
\newcommand{\KRR}{\textsc{Clique in}\text{ $r$}\textsc{-regular Graph}}
\newcommand{\KRRshort}{\textsc{K\text{$r$}R}}
\newcommand{\PCRR}{\textsc{Partial Complement  to $r$-Regular}\xspace}
\newcommand{\PCRRshort}{\textsc{PC}$r$\textsc{R}}
\newcommand{\PCTFshort}{\PCG}

\newcommand{\ie}{i.\,e.\@\xspace}
\newcommand{\eg}{e.\,g.\@\xspace}

\newcommandx{\defsimpleproblem}[3][]{
  \vspace{3mm}
  \noindent\fbox{
    \begin{minipage}{0.96\textwidth}
      \begin{tabular*}{\textwidth}{@{\extracolsep{\fill}}lr} {#1} \\ 
      \end{tabular*}
      {\textbf{Input:}} {#2}  \\
      {\textbf{Question:}} {#3}
    \end{minipage}
  }
  \vspace{4mm}
}

\title{Partial complementation of graphs\thanks{The first three authors have been supported by the Research Council of Norway via the projects ``CLASSIS'' and ``MULTIVAL". The fourth author has been supported by project ``DEMOGRAPH" (ANR-16-CE40-0028).}}

\author{%
Fedor V. Fomin\thanks{Department of Informatics, University of Bergen, Norway, \texttt{\{fedor.fomin, petr.golovach,torstein.stromme\}@ii.uib.no}}
\and 
\addtocounter{footnote}{-1}
Petr A. Golovach\footnotemark{}
\and 
\addtocounter{footnote}{-1}
Torstein J.\,F. Strømme \footnotemark{}
\and 
Dimitrios M. Thilikos \thanks{AlGCo project, CNRS, LIRMM, France, \texttt{sedthilk@thilikos.info}}~\thanks{Department of Mathematics National and Kapodistrian University of Athens, Greece.}
}

\begin{document}
\maketitle

\begin{abstract}
A \emph{partial complement} of the graph $G$ is a graph obtained from $G$ by complementing all the edges in one of its induced subgraphs. We study the following algorithmic question: for a given graph $G$ and graph class $\mathcal{G}$, is there a partial complement of $G$ which is in $\mathcal{G}$? We show that this problem can be solved in polynomial time for various choices of the graphs class $\mathcal{G}$, such as bipartite, degenerate, or cographs. We complement these results by proving that the problem is \classNP-complete when $\mathcal{G}$ is the class of $r$-regular graphs.
\end{abstract}

\section{Introduction}

One of the most important questions in graph theory concerns the efficiency of recognition of a graph class $\mathcal{G}$. For example, how fast we can decide whether a graph is chordal,
2-connected,
triangle-free,
of bounded treewidth,
bipartite,
$3$-colorable,
or excludes some fixed graph as a minor? In particular, the recent developments in parameterized algorithms are driven by the problems of recognizing of graph classes which do not differ up to a ``small disturbance'' from graph classes recognizable in polynomial time. The amount of disturbance is quantified in ``atomic'' operations required for modifying an input graph into the ``well-behaving'' graph class $\mathcal{G}$. The standard operations could be edge/vertex deletions, additions or edge contractions. 
Many problems in graph algorithms fall into this graph modification category: is it possible to add at most $k$ edges to make a graph $2$-edge connected or to make it chordal? Or is it possible to delete at most $k$ vertices such that the resulting graph has no edges or contains no cycles? 

%
%
%
%
%
%
%
%

A rich subclass of modification problems concerns edge editing problems. Here the ``atomic'' operation is the change of adjacency, \ie{} for a pair of vertices $u,v$, we can either add an edge $uv$ or delete the edge $uv$. For example, the \textsc{Cluster Editing} problem asks to transform an input graph into a cluster graph, that is a disjoint union of cliques, by flipping at most $k$ adjacency relations. 

Besides the basic edge editing, it is natural to consider problems where the set of removed and added edges should satisfy some structural constraints.  In particular, such problems were considered for \emph{complementation} problems. Recall that the \emph{complement} of a graph $G$ is a graph $H$ on the same vertices such that two distinct vertices of $H$ are adjacent if and only if they are not adjacent in $G$.
Seidel (see~\cite{Seidel74,Seidel76,Seidel81}) introduced the operation that is now known as the \emph{Seidel switch}. For a vertex $v$ of a graph $G$, this operation complements the adjacencies of $v$, that is, it removes the edges incident to $v$ and makes $v$ adjacent to the non-neighbors of $v$ in $G$. Respectively, for a set of vertices $U$, the Seidel switching, that is, the consecutive switching for the vertices of $U$, complements the adjacencies between $U$ and its complement $V(G)\setminus U$. 
The study of the algorithmic question whether it is possible to obtain a graph from a given graph class by the Seidel switch was initiated by  Ehrenfeucht et al.~\cite{EhrenfeuchtHHR98}.  
Further results were established in~\cite{JelinekJK16,JelinkovaK14,JelinkovaSHK11,KratochvilNZ92,Kratochvil03}.
Another important operation of this type is the \emph{local complementation}.  For a vertex $v$ of a graph $G$, the \emph{local complementation of $G$ at $v$} is the graph obtained from $G$ by replacing $G[N(v)]$ by its complement.  This operation plays crucial role in the definition of \emph{vertex-minors}~\cite{Oum05} and was investigated in this contest (see, e.g.~\cite{CourcelleO07,Oum17}). See also~\cite{Bouchet93,KAMINSKI20092747} for some algorithmic results concerning local complementations.

In this paper we study the \emph{partial complement} of a graph, which was introduced by Kami{\'n}ski, Lozin, and Milani{\v c} in \cite{KAMINSKI20092747} in their study of the clique-width of a graph. 
A \emph{partial complement} of a graph $G$ is a graph obtained from $G$ by complementing all the edges of one of its induced subgraphs.
 More formally, for a graph $G$ and $S\subseteq V(G)$, we define 
$G \oplus S$ as the graph with the vertex set $V(G)$ whose edge set is defined as follows: 
a pair of distinct vertices $u,v$ is an edge of $G \oplus S$ if and only if one of the following holds: 
\begin{itemize}
    \item $uv \in E(G) \land (u \notin S \lor v \notin S)$, or
    \item $uv \notin E(G) \land u \in S \land v \in S$. 
\end{itemize}
Thus when the set $S$ consists only of two vertices $\{u, v\}$, then the operation changes the adjacency between $u$ and $v$, and for a larger set $S$, 
$G \oplus S$ changes the adjacency relations for all pairs of vertices of $S$.

We say that a graph $H$ is a partial complement of the graph $G$ if $H$ is isomorphic to $G \oplus S$ for some $S\subseteq V(G)$. 
For a graph class $\mathcal{G}$ and a graph $G$, we say that there is a \emph{partial complement of $G$ to $\mathcal{G}$} if for some $S\subseteq V(G)$, we have $G \oplus S \in \mathcal{G}$. We denote by $\mathcal{G}^{(1)}$ the class of graphs such that its members can be partially complemented to $\mathcal{G}$.


Let $\mathcal{G}$ be a graph class. We consider the following generic algorithmic problem.

\defsimpleproblem{\PCG (\PCGshort{})}%
{A simple undirected graph $G$.}%
{Is there a partial complement of $G$ to $\mathcal{G}$?}

In other words, how difficult is it to recognize the class $\mathcal{G}^{(1)}$?
In this paper we show that  there are many well-known graph classes $\mathcal{G} $ such that 
$\mathcal{G}^{(1)}$ is recognizable in polynomial time. We show  that
\begin{itemize}
    \item \PCG is solvable in time $\bigO(f(n)\cdot n^4 + n^6)$ when $\mathcal{G}$ is a triangle-free graph class recognizable in time $f(n)$. For example, this implies that when  $\mathcal{G}$ is the class of bipartite graphs, the class $\mathcal{G}^{(1)}$ is recognizable in polynomial time.
    \item \PCG is solvable in time $f(n)\cdot n^{\bigO(1)}$ when $\mathcal{G}$ is a $d$-degenerate graph class recognizable in time $f(n)$.
    Thus when  $\mathcal{G}$ is the class of planar graphs, class of cubic graphs, class of graph of bounded treewidth, or class of $H$-minor free graphs, then the class $\mathcal{G}^{(1)}$ is recognizable in polynomial time.
    \item \PCG is solvable in polynomial time when $\mathcal{G}$ is a class of bounded clique-width  expressible in monadic second-order logic (with no edge set quantification). In particular, if $\mathcal{G}$ is the class of $P_4$-free graphs (cographs), then $\mathcal{G}^{(1)}$ is recognizable in polynomial time.
    \item \PCG is solvable in polynomial time when $\mathcal{G}$ can be described by a $2 \times 2$ $M$-partition matrix. Therefore $\mathcal{G}^{(1)}$ is recognizable in polynomial time when $\mathcal{G}$ is the class of split graphs, as they can be described by such a matrix. 
\end{itemize}

Nevertheless, there are cases when the problem is \classNP-hard. In particular, we prove that this holds  when $\mathcal{G}$ is the class of $r$-regular graphs.



\section{Partial complementation to triangle-free graph classes}\label{sec:triangle_free}

A triangle is a complete graph on three vertices. Many graph classes does not allow the triangle as a subgraph, for instance trees, forests, or graphs with large girth. In this paper 
we show that partial complementation to triangle-free graphs can be decided in polynomial time. 


More precisely, we show that if a graph class ${\cal G}$ can be recognized in polynomial time and it is triangle-free, then we can also solve \PCG{} in polynomial time.



 Our algorithm is constructive, and returns a \emph{solution} $S \subseteq V(G)$, that is a set $S$ such that $G\oplus S$ is in $\mathcal{G}$. We say that a solution \emph{hits} an edge $uv$ (or a non-edge $\overline{uv})$, if both $u$ and $v$ are contained in $S$.

Our algorithm considers each of the following cases. 
\begin{enumerate}
\item[($i$)]   There is a solution $S$ of size at most two. 
\item[($ii$)]  There is a solution $S$ containing two vertices that are non-adjacent in $G$.
\item[($iii$)] There is a solution $S$ such that it form a clique of size at least $3$ in $G$.
\item[($iv$)]   $G$ is a no-instance. 
\end{enumerate}

Case~($i$) can be resolved in polynomial time by brute-force, and thus we start from analyzing the structure of a solution in Case~($ii$). We need the following observation. 

\begin{observation} \label{tf:obs:small-independentset-in-image}
Let $  \mathcal{G}$ be a class of triangle-free graphs and let $G$ be an instance of \PCTFshort, where $S \subseteq V(G)$ is a valid solution. Then 
    \begin{enumerate}[a)]
        \item $G[S]$ does not contain an independent set of size 3, and
        \item for every triangle $\{u, v, w\} \subseteq V(G)$, at least two vertices are in $S$.
    \end{enumerate}
\end{observation}

\noindent Because all non-edges between vertices in $G[S]$ become edges in $G \oplus S$ and vice versa, whereas all (non-) edges with an endpoint outside $S$ remain untouched, we see that the observation holds.

Let us recall that a graph $G$ is a \emph{split graph} if its vertex set can be partitioned into $V(G)=C\cup I$, where $C$ is a clique and $I$ is an independent set.  Let us note that the vertex set of a split graph can have several \emph{split partitions}, i.e. partitions into a clique and independent set. However, the number of split partitions of an $n$-vertex split graphs is at most $n$.
The analysis of Case~($ii$) is based on the following lemma. 

\begin{lemma} \label{tf:lem:nonclique-solution-properties}
    Let $  \mathcal{G}$ be a class of triangle-free graphs and let $G$ be an instance of \PCTFshort{}. Let $S \subseteq V(G)$ be a valid solution which is not a clique, and let $u, v \in S$ be distinct vertices such that $uv \notin E(G)$. Then
    \begin{enumerate}[a)]
        \item the entire solution $S$ is a subset of the union of the closed neighborhoods of $u$ and $v$,  that is $S \subseteq N_G[u] \cup N_G[v]$;
        \item every common neighbor of $u$ and $v$ must be contained in the solution $S$, that is $N_G(u) \cap N_G(v) \subseteq S$;
        \item the graph $ G[N(u) \setminus N(v)]$ is a split graph. Moreover, $(N(u) \setminus N(v))\cap S$ is a clique  and  $(N(u) \setminus N(v))\setminus S$ is an independent set.
    \end{enumerate}
\end{lemma}

\begin{proof}
    We will prove each point separately, and in order.
    \begin{enumerate}[a)]
        \item Assume for the sake of contradiction that the solution $S$ contains a vertex $w \notin N_G[u] \cup N_G[v]$. But then $\{u, v, w\}$ is an independent set in $G$, which contradicts item a) of  Observation~\ref{tf:obs:small-independentset-in-image}.
        \item Assume for the sake of contradiction that the solution $S$ does not contain a vertex $w \in N_G(u) \cap N_G(v)$. Then the edges $uw$ and $vw$ will both be present in $G \oplus S$, as well as the edge $uv$. Together, these forms a triangle.
        \item We first claim that the solution $S$ is a vertex cover for $G[N(u) \setminus N(v)]$. If it was not, then there would exist an edge $u_1u_2$  of $G[N(u)\setminus N(v)]$ such that both  endpoints   $u_1, u_2\not\in S$, yet  $u_1, u_2$   would form a triangle with $u$ in $G\oplus S$, which would be a contradiction. Hence $(N(u) \setminus N(v))\setminus S$ is an independent set.
        Secondly, we claim that  $(N(u) \setminus N(v))\cap S$  forms a clique. If not, then there would exist
 $u_1, u_2 \in (N(u) \setminus N(v))\cap S$ which are nonadjacent. 
  In this case  $\{u_1, u_2, v\}$ is an independent set, which contradicts  item a) of Observation~\ref{tf:obs:small-independentset-in-image}. Taken together, these claims imply the last item of the lemma.
    \end{enumerate}
\end{proof}

\noindent We now move on to examine the structure of a solution for the third case, when there exists a solution which is a clique of size at least three.

\begin{lemma} \label{tf:lem:clique-solution-properties}
    Let $  \mathcal{G}$ be a class of triangle-free graphs and let $G$ 
  be an instance of \PCTFshort. Let $S \subseteq V(G)$ be a solution such that $|S| \geq 3$ and $G[S]$ is a clique. Let $u, v \in S$ be distinct. Then
    \begin{enumerate}[a)]
        \item the solution $S$  is contained in their common neighborhood, that is $S \subseteq N_G[u] \cap N_G[v]$, and
        \item the graph $G[N_G[u] \cap N_G[v]]$ is a split graph where $(N_G[u] \cap N_G[v])\setminus S$ is an independent set. 
    \end{enumerate}
\end{lemma}
\begin{proof}
    We prove each point separately, and in order.
    \begin{enumerate}[a)]
        \item Assume for the sake of contradiction that the solution $S$ contains a vertex $w$ which is not in the neighborhood of both $u$ and $v$. This contradicts that $S$ is a clique.
        \item We claim that $S$ is a vertex cover of $G[N_G[u] \cap N_G[v]]$. Because $S$ is also a clique, the statement of the lemma will then follow immediately. Assume for the sake of contradiction that $S$ is not a vertex cover. Then there exist an uncovered edge $w_1w_2$,  where $w_1, w_2 \in N_G[u] \cap N_G[v]$, and also $w_1, w_2 \notin S$. Since $\{u, w_1, w_2\}$ form a triangle, we have by b) of Observation~\ref{tf:obs:small-independentset-in-image} that at least two of these vertices are in $S$. That is a contradiction, so our claim holds. %
    \end{enumerate} %
\end{proof}
\noindent
We now have everything in place to present the algorithm.
\begin{algo}[{\PCG{} where $\mathcal{G}$ is triangle-free}] \label{tf:algo}
$ $ \newline{}
Input: An instance $G$ of \PCGshort{} where $\mathcal{G}$ is a triangle-free graph class recognizable in time $f(n)$ for some function $f$.
\newline{}
Output: A set $S \subseteq V(G)$ such that $G \oplus S$ is in $\mathcal{G}$, or a correct report that no such set exists.
\begin{enumerate}
    \item By brute force, check if there is a solution of size at most 2. If yes, return this solution.
    \item For every non-edge $\overline{uv}$ of $G$:
    \begin{enumerate}
        \item If either $G[N(u) \setminus N_G(v)]$ or $G[N_G(u) \setminus N_G(v)]$ is not a split graph, skip this iteration and try the next non-edge.
        \item Let $(I_u, C_u)$ and $(I_v, C_v)$ denote a split partition of $G[N_G(u)\setminus N_G(v)]$ and $G[N_G(v)\setminus N_G(u)]$ respectively. For each pair of split partitions $(I_u, C_u), (I_v, C_v)$:
        \begin{enumerate}
            \item Construct solution candidate $S' := \{u, v\} \cup (N_G(u) \cap N_G(v)) \cup C_u \cup C_v$
            \item If $G \oplus S'$ is a member of $\mathcal{G}$, return $S'$
        \end{enumerate}
    \end{enumerate}
    \item Find a triangle $\{x, y, z\}$ of $G$
    \item For each edge in the triangle $uv \in \{xy, xz, yz\}$:
    \begin{enumerate}
        \item If $G[N_G(u) \cap N_G(v)]$ is not a split graph, skip this iteration and try the next edge.
        \item For each possible split partition $(I, C)$ of $G[N_G(u) \cap N_G(v)]$:
        \begin{enumerate}
            \item Construct solution candidate $S' := \{u, v\} \cup C$
            \item If $G \oplus S'$ is a member of $\mathcal{G}$, return $S'$
        \end{enumerate}
    \end{enumerate}
    \item Return `\textsc{None}'
\end{enumerate}
\end{algo}

\begin{theorem}
Let $\mathcal{G}$ be a class of triangle-free graphs such that deciding whether an $n$-vertex graph is in $\mathcal{G}$ is solvable in time $f(n)$ for some function $f$. Then 
\PCG is solvable 
  in time $\bigO(n^6 + n^4 \cdot f(n))$.
\end{theorem}

\begin{proof}
    We will prove that Algorithm~\ref{tf:algo} is correct, and that its running time is $\bigO(n^4 \cdot (n^2 + f(n)))$. We begin by proving correctness. Step 1 is trivially correct. After Step 1 we can assume that any valid solution has size at least three, and we have handled Case~($i$) when there exists a solution of size at most two. We have the three cases left to consider: ($ii$) There exists a solution which hits a non-edge, ($iii$) 
    there is a solution $S$ such that in $G\oplus S$ vertices of $S$ form a clique of size at least $3$, and
     ($iv$) no solution exists.

    In the case that there exists a solution $S$ hitting a non-edge $uv$, we will at some point guess this non-edge in Step 2 of the algorithm. By Lemma~\ref{tf:lem:nonclique-solution-properties}, we have that both $G[N_G(u) \setminus N_G(v)]$ and $G[N_G(u) \setminus N_G(v)]$ are split graphs, so we do not miss the solution $S$ in Step 2a. Since we try every possible combinations of split partitions in Step 2b, we will by Lemma~\ref{tf:lem:nonclique-solution-properties} at some point construct $S'$ correctly such that $S' = S$.

    In the case that there exist only solutions which hits exactly a clique, we first find some triangle $\{x, y, z\}$ of $G$. It must exist, since a solution $S$ is a clique of size at least three. By Observation~\ref{tf:obs:small-independentset-in-image}b, at least two vertices of the triangle must be in the $S$. At some point in step 4 we guess these vertices correctly. By Lemma~\ref{tf:lem:clique-solution-properties}b we know that $G[N_G(u) \cap N_G(v)]$ is a split graph, so we will not miss $S$ in Step 4a. Since we try every split partition in Step 4b, we will by Lemma~\ref{tf:lem:clique-solution-properties} at some point construct $S'$ correctly such that $S' = S$.

    Lastly, in the case that there is no solution, we know that there neither exists a solution of size at most two, nor a solution which hits a non-edge, nor a solution which hits a clique of size at least three. Since these three cases exhaust the possibilities, we can correctly report that there is no solution when none was found in the previous steps.

    For the runtime, we start by observing that Step 1 takes time $\bigO(n^2 \cdot f(n))$. The sub-procedure of Step 2 is performed $\bigO(n^2)$ times, where step 2a takes time $\bigO(n \log n)$. The sub-procedure of Step 2b takes time at most $\bigO(n^2 + f(n))$, and it is performed at most $\bigO(n^2)$ times. In total, Step 2 will use no longer than $\bigO(n^4 \cdot (n^2 + f(n)))$ time. Step 3 is trivially done in time $\bigO(n^3)$. The sub-procedure of Step 4 is performed at most three times. Step 4a is done in $\bigO(n \log n)$ time, and step 4b is done in $\bigO(n \cdot (n^2 + f(n))$ time, which also becomes the asymptotic runtime of the entire step 4. The worst running time among these steps is Step 2, and as such the runtime of Algorithm~\ref{tf:algo} is $\bigO(n^4 \cdot (n^2 + f(n)))$.
\end{proof} 



\section{Complement  to degenerate graphs}\label{sec:degenerated}

 For $d>0$, we say that a graph $G$ is $d$-degenerate, if every induced (not necessarily proper) subgraph of $G$ has a vertex of degree at most $d$. For example, trees are $1$-degenerate, while planar graphs are $5$-degenerate.  
 
\begin{theorem} \label{thm:ddeg}
Let $\mathcal{G}$ be a  class of $d$-degenerate graphs such that deciding 
 whether an $n$-vertex graph is in $\mathcal{G}$  is solvable in time $f(n)$ for some function $f$. 
 Then 
\PCG is solvable in time $f(n) \cdot n^{2^{\bigO(d)}}$. 
 
\end{theorem}
\begin{proof}Let $G$ be an $n$-vertex graph.  We are looking for a vertex subset  
 $S$ of $G$ such that $G\oplus S\in \mathcal{G}$. 
 
 We start from trying all vertex subsets of $G$ of size at most $2d$ as a candidate for $S$.
 Thus, in time $\bigO(n^{2d}\cdot f(n))$ we either find a solution or conclude that a solution, if it exists, should be of size more than  $2d$.
 
Now we assume that $|S|>2d$. We try all subsets of $V(G)$ of size $2d+1$. Then if $G$ can be complemented to $\mathcal{G}$, at least one of these sets, say $X$, is a subset of $S$. In total, we enumerate ${\binom{n}{2d+1}}$ sets. 

  First we consider the set $Y$ of  all vertices in $ V(G)\setminus X$ with at least $d+1$ neighbors in $X$. The observation here is that most vertices from $Y$ are in $S$. More precisely, if more than 
  \[
 \alpha= \binom{|X|}{d+1} \cdot d+1=  \binom{2d+1}{d+1} \cdot d+1  
  \]
  vertices of $Y$ are not in $S$, then $G\oplus S$ contains a complete bipartite graph $G_{d+1,d+1}$ as a subgraph, and hence $G\oplus S$ is not $d$-degenerate. Thus, we make at most 
 $
  \binom{n}{\alpha}
$
  guesses on which subset of $Y$ is in $S$. 
  
  Similarly, when we consider the set $Z$ of  all vertices from $ V(G)\setminus X$ with at most $d$ neighbors in $X$, we have that at most $\alpha$ of vertices from $Z$ could belong to $S$. Since $V(G)=X\cup Y\cup Z$, if there is a solution $S$, it will be found in at least one from  
  \[
  \binom{n}{2d+1} \cdot  \alpha^2=n^{2^{\bigO(d)}}
  \] 
  of the guesses.
  Since for each set $S$ we can check in time $f(n)$ whether  $G\oplus S\in \mathcal{G}$, this concludes the proof.
\end{proof}



\section{Complement to M-partition}

Many graph classes can be defined by whether it is possible to partition the vertices of graphs in the class such that certain internal and external edge requirements of the parts are met. For instance, a complete bipartite graph is one which can be partitioned into two sets such that every edge between the two sets is present (external requirement), and no edge exists within any of the partitions (internal requirements). Other examples are split graphs and $k$-colorable graphs. 
Feder et al.~\cite{feder2003list} formalized such partition properties of graph classes  by making use of a symmetric matrix over $\{0, 1, \star\}$, called an \emph{$M$-partition}. 

\begin{definition}[{$M$-partition}]
    For a $k \times k$ matrix $M$, we say that a graph $G$ belongs to the graph class $\mathcal{G}_M$ if its vertices can be partitioned into $k$ (possibly empty) sets $X_1, X_2, \dots, X_k$ such that, for every $i \in [k]$, if

    \begin{itemize}
        \item  $M[i,i] = 1$, then $X_i$ is a clique and if  $M[i,i] = 0$, then $X_i$ is an  independent set, and 
            \end{itemize}
    for every $i, j \in [k]$, $i\neq j$,  
    \begin{itemize}
        \item if $M[i,j] = 1$, then every vertex of $X_i$ is adjacent to all vertices of $X_j$, 
        \item if $M[i,j] = 0$, then there is no edges between $X_i$  and $X_j$.
    \end{itemize}
\end{definition}
\noindent Note that if $M[i,j] = \star$, then there is no restriction on the edges between vertices from $X_i$ and $X_j$.

For example, for matrix   
\[M=\left(\begin{array}{cc}0 & \star \\ \star & 0\end{array}\right)\]
the corresponding class of graphs is the class of bipartite graphs, while   matrix 
\[M=\left(\begin{array}{cc}0 & \star \\ \star& 1\end{array}\right)\]
identifies the class of split graphs. 

In this section we prove the following theorem. 

\begin{theorem}\label{thm:Mpartition} 
    Let $\mathcal{G}= \mathcal{G}_M$ be a graph class described by an $M$-partition matrix of size $2 \times 2$. Then \PCG{} is solvable in polynomial time. 
\end{theorem}

In particular, Theorem~\ref{thm:Mpartition} yields  
 polynomial algorithms for \PCG{} when $\mathcal{G}$ is the class of split graphs or (complete) bipartite graphs. The proof of our theorem is based on the following beautiful dichotomy result of Feder et al.~\cite{feder2003list} on the recognition of classes  $\mathcal{G}_M$ described by $4\times 4$ matrices.  

%
%
\begin{proposition}[{\cite[Corollary 6.3]{feder2003list}}] \label{mpart:prop:mpart44}
    Suppose $M$ is a symmetric matrix over $\{0, 1, \star\}$ of size $k = 4$. Then the recognition problem for $\mathcal{G}_M$ is 
    \begin{itemize}
        \item NP-complete when $M$ contains the matrix for 3-coloring or its complement, and no diagonal entry is $\star$.
        \item Polynomial time solvable otherwise.
    \end{itemize}
\end{proposition}

\begin{lemma}\label{lem:Mpart}
    Let $M$ be a symmetric $k\times k$ matrix giving rise to the graph class $\mathcal{G}_M = \mathcal{G}$. Then there exists a $2k \times 2k$ matrix $M'$ such that for any input $G$ to \PCG{}, it is a yes-instance if and only if $G$ belongs to $\mathcal{G}_{M'}$.
\end{lemma}
\begin{proof}
    Given $M$, we construct a matrix $M'$ in linear time. We let $M'$ be a matrix of dimension $2k \times 2k$, where entry $M'[i,j]$ is defined as $M[\lceil\frac{i}{2}\rceil,\lceil\frac{j}{2}\rceil]$ if at least one of $i,j$ is even, and $\neg{M[\frac{i+1}{2},\frac{j+1}{2}]}$ if $i,j$ are both odd. Here, $\neg{1} = 0$, $\neg{0} = 1$, and $\neg{\star} = \star$.  For example, for matrix 
\[M=\left(\begin{array}{cc}0 & \star \\ \star& 1\end{array}\right)\]
    the  above  construction results in  
    \[ M'=  
    \left(\begin{array}{cccc}1 & 0 & \star & \star \\
                                        0 & 0 & \star & \star \\  
                                        \star & \star & 0 & 1 \\ 
                                        \star & \star & 1 & 1
                                        \end{array}\right) .\]

    We prove the two directions separately.
    
    ($\implies$) Assume there is a partial complementation $G\oplus S$ into $\mathcal{G}_M$. 
 Let $X_1, X_2, \dots, X_k$ be an $M$-partition of  $G\oplus S$. We define partition $X'_1, X'_2, \dots, X'_{2k}$ of $G$ as follows.
 For every  vertex $v \in X_i$, $1\leq i \leq k$, we assign $v$ to $X'_{2i-1}$ if $v\in S$ and to  $X'_{2i}$ otherwise.
 
    We now show that every edge   of $G$ respects the requirements of $M'$. Let $uv \in E(G)$ be an edge, and let 
    $u\in X_i$ and $v\in X_j$.  If at least one vertex from $\{u,v\}$, say $v$  is not in $S$, then $uv$ is also an edge in $G\oplus S$, thus $M[i,j] \neq 0$. 
    Since $v\not\in S$, it belongs to set $v\in X'_{2j}$. Vertex $u$ is assigned  to set $X'_{\ell}$, where $\ell$ is either $2i$ or $2i-1$, depending whether $u$ belongs to $S$ or not. But because $2j$ is even irrespectively of $\ell$, $M'[\ell, 2j]=M[i,j]\neq 0$. 
    
%
%

    Now consider the case when both $u,v \in S$. Then the edge does not persist after the partial complementation by $S$, and thus $M[i ,j ] \neq 1$. We further know that $u$ is assigned to   $X'_{2i-1}$ and  $v$ to  $X'_{2j-1}$.  Both $2i-1$ and $2j-1$ are odd, and 
 by the construction of $M'$, we have that  $M'[2i-1,2j-1] \neq 0$, and again the edge $uv$ respects $M'$. An analogous argument shows that also all non-edges respect $M'$.

    ($\impliedby$) Assume that there is a partition $X'_1, X'_2, \dots, X'_{2k}$ of $G$ according to $M'$. Let the set $S$ consist of all vertices in  odd-indexed parts of the partition. We now show that $G\oplus S$ can be partitioned according to $M$. We define partition 
   $X_1, X_2, \dots, X_k$ by assigning   
 each vertex $u\in X'_i$ to  $X_{\lceil\frac{i}{2}\rceil}$. It remains to show that 
 $X_1, X_2, \dots, X_k$ is an $M$-partition of 
  $G\oplus S$.

 Let $u\in X_i$, $v\in X_j$. 
   Suppose first that  $uv\in E(G\oplus S)$. 
    If at least one of $u,v$ is not in $S$, we assume without loss of generality that $v \notin S$. 
    Then $uv\in E(G)$ and   $v\in X'_{2j}$. For vertex $u\in X'_{\ell}$, irrespectively, whether $\ell$ is  $2i$ or $2i-1$, we have that  $M'[\ell, 2j]=M[i,j]\neq 0$. 
But then  $M[i,j] \neq 0$. 
    Otherwise we have $u,v \in S$. Then $uv$ is a non-edge in $G$, and thus  $M'[2i-1,2j-1] \neq 1$. But by the construction of $M'$, we have  that $M[i ,j] \neq 0$, and there is no violation of $M$. An analogous argument shows that if $u$ and $v$ are not adjacent in  $G\oplus S$, it holds that $M[i,j] \neq 1$. Thus 
    $X_1, X_2, \dots, X_k$ is an $M$-partition of 
  $G\oplus S$,  which concludes the proof.
\end{proof}

Now we are ready to prove Theorem~\ref{thm:Mpartition}.
\begin{proof}[Proof of Theorem~\ref{thm:Mpartition}]
For a given matrix $M$, we use  Lemma~\ref{lem:Mpart} to construct a matrix $M'$. Let us note that by the construction of matrix $M'$, for 
every $2 \times 2$ matrix $M$ we have that  matrix $M'$ has at most two $1$'s and at most two $0$'s along the diagonal. Then by Proposition~\ref{mpart:prop:mpart44}, the recognition of whether $G$ admits  $M'$-partition is in P. Thus by Lemma~\ref{lem:Mpart},  \PCG{} is solvable in polynomial time
\end{proof}



\section{Partial complementation to graph classes of bounded clique-width}

We show that \PCG{} can be solved in polynomial time when $\mathcal{G}$ has bounded clique-width and can be expressed by an \MSOone{} property. 
We refer to the book~\cite{Courcelle:2012:GSM:2414243} for the basic definitions.
  We will use the following result of Hlin{\v e}n{\'y} and Oum~\cite{HlinenyO08}.
 
 \begin{proposition}[\cite{HlinenyO08}] \label{prop:cliquewappro}     
 There is an algorithm that for every integer $k$ and graph $G$ in time  $O(|V(G)|^3)$ either  computes a $(2^{k+1}
- 1)$ expression for a graph $G$ or correctly concludes that the clique-width of $G$ is  more than $k$.
 \end{proposition}

Note that  the algorithm of Hlin{\v e}n{\'y} and Oum 
only approximates the clique-width but does not provide an
algorithm to construct an optimal $k$-expression tree for a graph
$G$ of clique-width at most $k$. But this approximation is usually
sufficient for algorithmic purposes.


 Courcelle, Makowsky and Rotics~\cite{Courcelle2000} proved  that every graph property that can be expressed in \MSOone{} can be recognized in linear time for graphs of bounded clique-width when given a $k$-expression. 

\begin{proposition}[{\cite[{Theorem~4}]{Courcelle2000}}]\label{rw:prop:courcelle} Let $\mathcal{G}$ be some class of graphs of clique-width at most $k$ such that for each graph $G \in \mathcal{G}$, a corresponding $k$-expression
can be found in $\bigO(f(n, m))$ time. Then every \MSOone{} property on $\mathcal{G}$ can be recognized in time $\bigO(f(n, m)+n)$.
\end{proposition}

The nice property of graphs with bounded clique-width is that their partial complementation is also bounded. In particular, Kami{\'n}ski, Lozin, and Milani{\v c} in \cite{KAMINSKI20092747} observed that if $G$ is a graph of clique-width $k$, then any partial complementation of $G$ is of clique-width at most $g(k)$ for some computable function $g$. For  completeness, we provide a more accurate upper bound.


\begin{lemma}\label{lem:upper-cw}
Let $G$ be a graph,  $S\subseteq V(G)$. Then $\cwd(G\oplus S)\leq 3\cwd(G)$.
\end{lemma}

\begin{proof}
Let $\cwd(G)=k$.
To show the bound, it is more convenient to use expression trees instead of $k$-expressions.
An {\em expression tree} of a graph $G$ is a rooted tree $T$ with nodes of four types $i$, $\dot{\cup}$, $\eta$ and $\rho$:
\begin{itemize}
 \item \emph{Introduce nodes $i(v)$} are leaves of $T$  corresponding to initial $i$-graphs with vertices $v$  labeled by $i$.
 \item\emph{Union node $\dot{\cup}$} stands for a disjoint union of graphs associated with its children.
 \item \emph{Relabel node $\rho_{i\to j}$} has one child and is associated with the $k$-graph obtained by  applying of the relabeling operation to  the graph corresponding to its child.
 \item \emph{Join node $\eta_{i,j}$} has one child and is associated with the $k$-graph resulting by applying the  join operation to the graph corresponding to its child.
 \item The graph $G$ is isomorphic to the graph associated with the root of $T$ (with all labels removed).
\end{itemize}
The {\em width} of the tree $T$ is the number of different labels appearing in $T$. If $G$ is of clique-width $k$, then by parsing the corresponding $k$-expression, one can construct an expression tree of width $k$ and, vise versa, given an expression tree of width $k$, it is straightforward to construct a $k$-expression.
Throughout the proof we call the elements of $V(T)$ \emph{nodes} to distinguish them from the vertices of $G$.
 Given a node $x$ of an expression tree, $T_x$ denotes the subtree of $T$ rooted in $x$ and the graph $G_x$ represents the $k$-graph formed by $T_x$.

An expression tree $T$ is \emph{irredundant} if for any join node $\eta_{i,j}$, the vertices labeled by $i$ and $j$ are not adjacent in the graph associated with its child.  
It was shown by Courcelle and Olariu~\cite{CourcelleO00} that every expression tree $T$ of $G$ can be transformed into an irredundant expression tree $T'$ of the same width in time linear in the size of $T$. 

Let $T$ be an irredundant expression tree of $G$ with the width $k$ rooted in $r$. We construct the expression tree $T'$ for $G'=G\oplus S$ by modifying $T$. 

Recall that the vertices of the graphs $G_x$ for $x\in V(T)$ are labeled $1,\ldots,k$. We introduce three groups of distinct labels $\alpha_1,\ldots,\alpha_k$, $\beta_1,\ldots,\beta_k$ and $\gamma_1,\ldots,\gamma_k$. The labels $\alpha_1,\ldots,\alpha_k$ and $\beta_1,\ldots,\beta_k$ correspond the the labels $1,\ldots,k$ for the vertices in $S$ and $V(G)\setminus S$ respectively. The labels $\gamma_1,\ldots,\gamma_k$ are auxiliary.
Then for every node $x$ of $T$ we construct $T_x'$ using $T_x$ starting the process from the leaves. We denote by $G_x'$ the $k$-graph corresponding to the root $x$ of $T_x'$.

For every introduce node $i(v)$, we construct an introduce node $\alpha_i(v)$ if $v\in S$ and an introduce node $\beta_i(v)$ if $v\notin S$. Let $x$ be a non-leaf node of $T$ and assume that we already constructed the modified expression trees of the children of $x$. 

Let $x$ be a union node $\dot{\cup}$ of $T$ and let $y$ and $z$ be  its children.  

We construct $k$ relabel nodes $\rho_{\alpha_i,\gamma_i}$ for $i\in \{1,\ldots,k\}$ that form a path, make one end-node of the path adjacent to $y$ in $T_y'$ and make the other end-node denoted by $y'$ the root of $T_{y'}'$ constructed from $T_y'$. Notice that in the corresponding graph $G_{y'}'$ all the vertices of $S$ are now labeled by $\gamma_1,\ldots,\gamma_k$ instead of $\alpha_1,\ldots,\alpha_k$. 

Next, we construct a union node $\dot{\cup}$ denoted by $x^{(1)}$ with the children $y'$ and $z$. This way we construct the disjoint union of $G_{y'}'$ and $G_z'$. 

Notice that the vertices that are labeled by the same label in $G_y$ and $G_z$ are not adjacent in $G$. Respectively, we should make the vertices of $V(G_x)\cap S$ and $V(G_y)\cap S$ with the same label adjacent in $G'$. We achieve it by adding $k$ join nodes $\eta_{\alpha_i,\gamma_i}$ for $i\in \{1,\ldots,k\}$, forming a path  out of them and making one end-node of the path adjacent to $x^{(1)}$. We declare the other end-node of the path denoted by $x^{(2)}$ the new root. 

Observe now that for the set of vertices $Y_i$ of $G_y$ labeled $i$ and the set of vertices $Z_j$ of $G_z$ labeled by $j$ where  $i,j\in \{1,\ldots,k\}$ are distinct, it holds that  the vertices of $Y_i$ and $Z_j$ are either pairwise adjacent in $G$ or pairwise nonadjacent. Respectively, on this stage of construction we ensure that if the vertices of $Y_i$ are not adjacent to the vertices of $Z_j$, then the vertices of $Y_i\cap S$ and $Z_j\cap S$ are made adjacent in $G'$. To do it, for every two distinct $i,j\in\{1,\ldots,k\}$ such that the vertices of $Y_i$ and $Z_j$ are not adjacent in $G$, construct a new join node   $\eta_{\gamma_i,\alpha_j}$ and form a path with all these nodes whose one end-node is adjacent to $x^{(2)}$ and the other end-node $x^{(3)}$ is the new root (we assume that $x^{(3)}=x^{(2)}$ if have no new constricted nodes).

Finally, we add $k$ relabel nodes $\rho_{\gamma_i,\alpha_i}$ for $i\in \{1,\ldots,k\}$ that form a path, make one end-node of the path adjacent to $x^{(3)}$ and make the other end-node denoted by $x$ the root of the obtained $T_{x}'$. Clearly,  all the vertices of $S$ in $G_x'$ are labeled by  $\alpha_1,\ldots,\alpha_k$. 

Let $x$ be a relabel node $\rho_{i\to j}$ of $T$ and let $y$ be  its child.  We construct two relabel nodes $\rho_{\alpha_i\to \alpha_j}$ and $\rho_{\beta_i\to \beta_j}$ denoted by $x$ and $x'$ respectively. We make $x'$ the child of $x$ and we make  the root $y$ of $T_y'$ the child of $x'$. 

Now, let  $x$ be a join node $\eta_{i\to j}$ of $T$ and let $y$ be  its child. Recall that $T$ is irredundant, that is, the vertices labeled by $i$ and $j$ in $G_y$ are not adjacent. Clearly, we should avoid making adjacent the vertices in $S$ in the construction of $G'$. We do it by constructing three new join nodes  $\eta_{\alpha_i\to \beta_j}$, $\eta_{\alpha_j\to \beta_i}$ and $\eta_{\beta_i\to \beta_j}$ denoted by $x,x',x''$ respectively. We make $x'$ the child of $x$, $x''$ the child of $x'$ and the node $y$ of $T_y'$ is made the child of $x''$.

This completes the description of the construction of $T'$. Using standard inductive arguments, it is straightforward to verify that $G'$ is isomorphic to the graph associated with the root of $T'$, that is, $\cwd(G')\leq 3k$. 
\end{proof}


\begin{lemma} \label{rw:lem:msoone}
    Let $\varphi$ be an \MSOone{} property describing the graph class $\mathcal{G}$. Then there exists an \MSOone{} property $\phi$ describing the graph class $\mathcal{G}^{(1)}$ of size $|\phi| \in \bigO(|\varphi|)$.
\end{lemma}
\begin{proof}
    We will construct $\phi$ from $\varphi$ in the following way: We start by prepending $\exists S \subseteq V(G)$. Then for each assessment of the existence of an edge in $\varphi$, say $uv \in E(G)$, replace that term with $((u \notin S \lor v \notin S) \land uv \in E(G)) \lor (u \in S \land v \in S \land uv \notin E(G))$. Symmetrically, for each assessment of the non-existence of an edge $uv \notin E(G)$, replace that term with $((u \notin S \lor v \notin S) \land uv \notin E(G)) \lor (u \in S \land v \in S \land uv \in E(G))$.

    We observe that if $\varphi$ is satisfiable for some graph $G$, then for every $S \subseteq V(G)$, the partial complementation $G\oplus S$ will yield a satisfying assignment to $\phi$. Conversely, if $\phi$ is satisfiable for a graph $G$, then there exist some $S$ such that $\varphi$ is satisfied for $G \oplus S$. For the size, we note that each existence check for edges blows up by a constant factor.
\end{proof}

We are ready to prove the main result of this section. 
\begin{theorem}\label{thm:clique-width}
    Let $\mathcal{G}$ be a graph class expressible in \MSOone{} which has bounded clique-width. 
    Then \PCG{} is solvable in polynomial time.
\end{theorem}

\begin{proof}
  Let $\varphi$ be the \MSOone{} formula which describes $\mathcal{G}$, and let $G$ be an $n$-vertex input graph. 
 We apply    
 Proposition~\ref{prop:cliquewappro}  for $G$ and in   time $O(n^3)$ either obtain   a  $(2^{3k+1}- 1)$ expression for $G$ or conclude that the clique-width of $G$ is more than $3k$. 
In the latter case,
  by Lemma~\ref{lem:upper-cw},  $G$ cannot be partially complemented to 
 $\mathcal{G}$.
%
%
%
   
     We then obtain an \MSOone{} formula  $\phi$ from Lemma~\ref{rw:lem:msoone}, and apply  Proposition~\ref{rw:prop:courcelle}, which works in   time $f(k, \phi) \cdot n$ for some function $f$. In total, the runtime of the algorithm is  $f(k, \phi) \cdot n + n^3$.
\end{proof}

  We remark that if clique-width expression is provided along with the input graphs, and $\mathcal{G}$ can be expressed in \MSOone{}, then there is a linear time algorithm for \PCG. This follows directly from Lemma~\ref{rw:lem:msoone} and Proposition~\ref{rw:prop:courcelle}.

  Theorem~\ref{thm:clique-width} implies that for every class of graphs $\mathcal{G}$ of bounded clique-width characterized by a finite set of finite forbidden induced subgraphs, \eg $P_4$-free graphs (also known as cographs) or classes of graphs discussed in \cite{BlancheD0LPZ17}, the  \PCG  problem is solvable in polynomial time. However, Theorem~\ref{thm:clique-width}  does not imply that  \PCG is solvable in polynomial time for  $\mathcal{G}$ being of the class of graphs having clique-width at most $k$. This is because such a class   $\mathcal{G}$  cannot be described by  \MSOone{}. Interestingly, for the related class  $\mathcal{G}$   of graphs of bounded \emph{rank-width} (see~\cite{CourcelleO00} for the definition) at most $k$, the result of Oum and Courcelle  \cite{CourcelleO07} combined with Theorem~\ref{thm:clique-width}
implies that  \PCG is solvable in polynomial time.
  
%


 \section{Hardness of partial complementation  to r-regular graphs}\label{sec:regulargraphs}

\noindent
Let us remind that a graph $G$ is $r$-regular if all its vertices are of degree $r$. 
We consider the following restricted version of \PCG{}.

\defsimpleproblem{\PCRR{} (\PCRRshort{})}{A simple undirected graph $G$, a positive integer $r$.}{Does there exist a vertex set $S \subseteq V(G)$ such that $G \oplus S$ is $r$-regular?}

\noindent
In this section, we show that \PCRR{} is NP-complete by a reduction from \KRR{}. 

\defsimpleproblem{\KRR{} (\KRRshort{})}{A simple undirected graph $G$ which is $r$-regular, a positive integer $k$.}{Does $G$ contain a clique on $k$ vertices?}


\noindent We will need the following well-known proposition.
\begin{proposition}[{\cite{GareyJ79}}] \label{rr:prop:krr-hard}
    \KRR{} is NP-complete.
\end{proposition}

\begin{theorem}\label{thm_NP_regular}
\PCRR 
is NP-complete. 
\end{theorem}

\begin{proof}
\noindent
We begin by defining a gadget which we will use in the reduction. 
\begin{figure}[h]
    \centering
    \includegraphics[scale=.4]{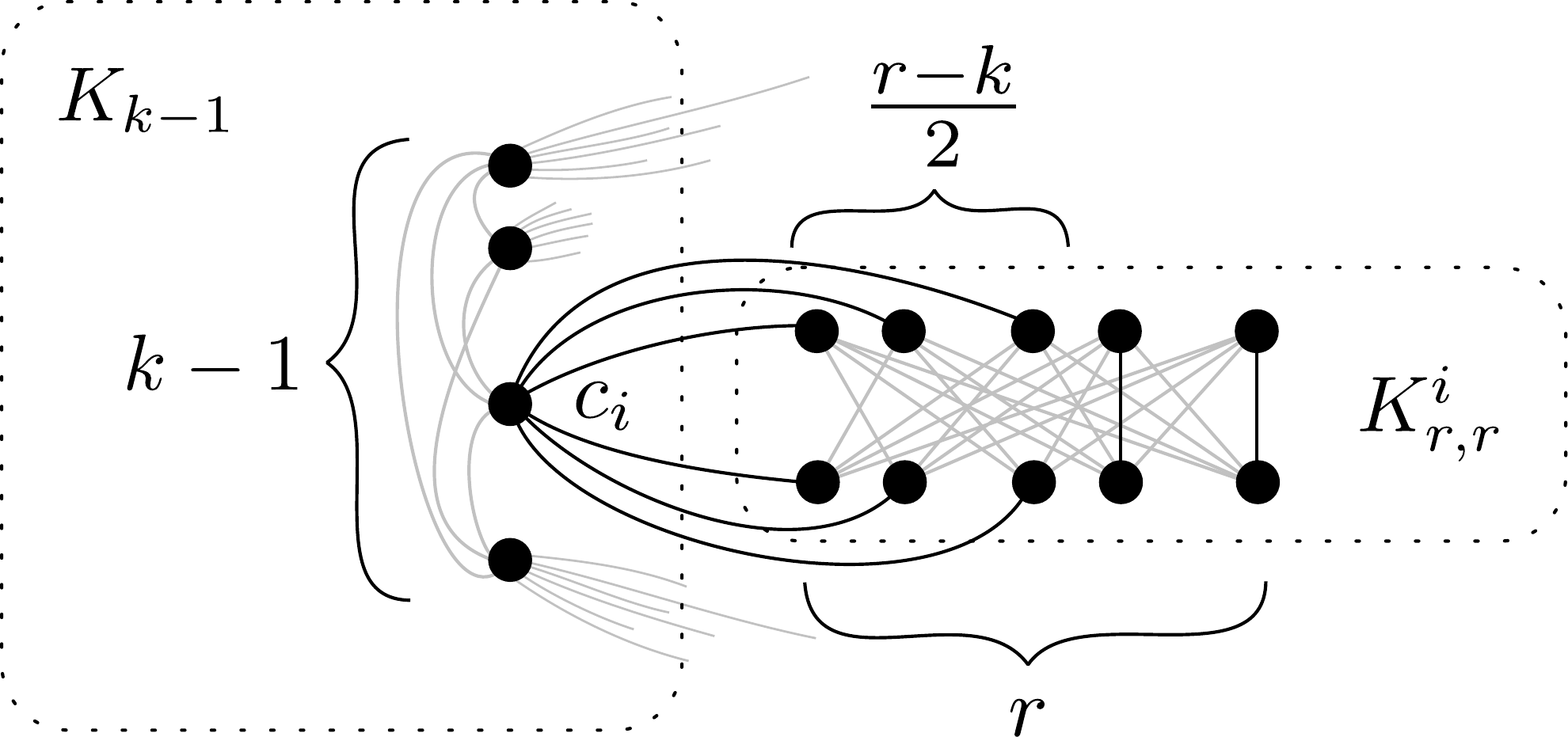}
    \caption{The graph $\textsc{gdg}_{k, r}$ is built of $k$ parts, namely a clique $K_{k-1}$, and $k-1$ complete bipartite graphs $K^1_{r,r}, \ldots, K^{k-1}_{r,r}$ with some rewiring.}\label{rr:fig:gdg}
\end{figure}
For integers $r>k$ such that $r-k$ is even, we build the graph $\textsc{gdg}_{k, r}$ as follows. Initially, we let $\textsc{gdg}_{k, r}$ consist of one clique on $k-1$ vertices, as well as $k-1$ distinct copies of $K_{r,r}$. These are all the vertices of the gadget, which is a total of $(k-1) + 2r \cdot (k-1)$ vertices. We denote the vertices of the clique $c_1, c_2, \ldots ,c_{k-1}$, and we let the complete bipartite graphs be denoted by $K^1_{r,r}, K^2_{r,r}, \ldots, K^{k-1}_{r,r}$. For a bipartite graph $K^i_{r,r}$, let the vertices of the two parts be denoted by $a^i_1, a^i_2, \ldots, a^i_r$ and $b^i_1, b^i_2, \ldots , b^i_r$ respectively.

We will now do some rewiring of the edges to complete the construction of $\textsc{gdg}_{k, r}$. Recall that $r-k$ is even and positive. For each vertex $c_i$ of the clique, add one edge from $c_i$ to each of $a^i_1, a^i_2, \ldots ,a^i_{\frac{r-k}{2}}$. Similarly, add an edge from $c_i$ to each of $b^i_1, b^i_2, \ldots ,b^i_{\frac{r-k}{2}}$. Now remove the edges $a^i_1b^i_1, a^i_2b^i_2, \ldots ,a^i_{\frac{r-k}{2}}b^i_{\frac{r-k}{2}}$. Once this is done for every $i \in [k-1]$, the construction is complete. See Figure~\ref{rr:fig:gdg}.

We observe the following property of vertices $a^i_j$, $b^i_j$, and $c_i$ of $\textsc{gdg}_{k, r}$.
\begin{observation} For every $i \in [k-1]$ and $j\in [r]$, it holds that the degrees of $a^i_j$ and $b^i_j$ in $\textsc{gdg}_{k, r}$ are both exactly $r$, whereas the degree of $c_i$ is $r-1$.
\end{observation}

\noindent
We are now ready to prove that \KRR{} is many-one reducible to
\PCRR{}.

\begin{algo}[Reduction \KRRshort{} to \PCRRshort{}] \label{rr:alg:reduction} 
$ $ \newline
    Input: An instance $(G, k)$ of \KRRshort{}.  \newline
 Output: An instance $(G', r)$ of \PCRRshort{} such that it is a yes-instance if and only if $(G, k)$ is a yes-instance of \KRRshort{}.
    \begin{enumerate}
        \item If $k < 7$ or $k \geq r$, solve the instance of \KRRshort{} by brute force. If it is a yes-instance, return a trivial yes-instance to \PCRRshort{}, if it is a no-instance, return a trivial no-instance to \PCRRshort{}.
        \item If $r-k$ is odd, modify $G$ by taking two copies of $G$ which are joined by a perfect matching between corresponding vertices. Then $r$ increase by one, whereas $k$ remains the same.
        \item Construct the graph $G'$ by taking the disjoint union of $G$ and the gadget $\textsc{gdg}_{k, r}$. Here, $r$ denotes the regularity of $G$ after step 2 is performed. Return $(G', r)$.

    \end{enumerate}
 \end{algo}

\noindent
Let $n=|V(G)|$.
We observe that the number of vertices in the returned instance is at most $2n + (k-1) + 2r \cdot (k-1)$, which is $\bigO(n^2)$. The running time of the algorithm is $\bigO(n^7)$ and thus is polynomial.

The correction of the reduction follows from the following two lemmata.

\begin{lemma}\label{lem_NPhardness1}
    Let $(G, k)$ be the input of Algorithm~\ref{rr:alg:reduction}, and let $(G', r)$ be the returned result. If $(G, k)$ is a yes-instance to \KRR{}, then $(G', r)$ is a yes-instance of \PCRR{}.
\end{lemma}

\begin{proof}
    Let $C \subseteq V(G)$ be a clique of size $k$ in $G$. If the clique is found in step 1, then $(G', r)$ is a trivial yes-instance, so the claim holds. Thus, we can assume that the graph $G'$ was constructed in step 3. If $G$ was altered in step 2, we let $C$ be the clique in one of the two copies that was created. Let $S \subseteq V(G')$ consist of the vertices of $C$ as well as the vertices of the clique $K_{k-1}$ of the gadget $\textsc{gdg}_{k, r}$. We claim that $S$ is a valid solution to $(G', r)$.

    We show that $G' \oplus S$ is $r$-regular. Any vertex not in $S$ will have the same number of neighbors as it had in $G'$. Since the only vertices that weren't originally of degree $r$ were those in the clique $K_{k-1}$, all vertices outside $S$ also have degree $r$ in $G' \oplus S$. What remains is to examine the degrees of vertices of $C$ and of $K_{k-1}$.

    Let $c_i$ be a vertex of $K_{k-1}$ in $G'$. Then $c_i$ lost its $k-2$ neighbors from $K_{k-1}$, gained $k$ neighbors from $C$, and kept $r-k$ neighbors in $K^i_{r,r}$. We see that its new neighborhood has size $k + r - k = r$.

    Let $u \in C$ be a vertex of the clique from $G$. Then $u$ lost $k-1$ neighbors from $C$, gained $k-1$ neighbors from $K_{k-1}$, and kept $r-(k-1)$ neighbors from $G - C$. In total, $u$ will have $r-(k-1)+(k-1) = r$ neighbors in $G' \oplus S$. Since every vertex of $G' \oplus S$ has degree $r$, it is $r$-regular, and thus $(G', r)$ is a yes-instance.
\end{proof}

\begin{lemma}\label{lem_NPhardness2}
    Let $(G, k)$ be the input of Algorithm~\ref{rr:alg:reduction}, and let $(G', r)$ be the returned result. If $(G', r)$ is a yes-instance to \PCRR{}, then $(G, k)$ is a yes-instance of \KRR{}.
\end{lemma}

\begin{proof}
    Let $S \subseteq V(G')$ be a solution witnessing that $(G', r)$ is a yes-instance. If $(G', r)$ was the trivial yes-instance returned in step 1 of Algorithm~\ref{rr:alg:reduction}, the statement trivially holds. Going forward we may thus assume $(G', r)$ was returned in step 3, and that $k \geq 7$.

    Because $G' \oplus S$ is $r$-regular, it must be the case that every vertex of $K_{k-1}$ is in $S$, since by construction these are the vertices which do not have degree $r$ in $G'$.

    We claim that $|S| = 2k - 1$, and moreover, that no neighbor of $K_{k-1}$ is in $S$. To show this, we let $p = |S \setminus K_{k-1}|$, and proceed to show that $p = k$. Towards this end, consider a vertex $c_i \in K_{k-1}$. This vertex has some number of neighbors in $S\setminus K_{k-1}$, denoted $x_i = |N_{G'}(c_i) \cap (S \setminus K_{k-1})|$. We know that $c_i$ has $r$ neighbors in $G' \oplus S$. Let us count them: Some neighbors are preserved by the partial complementation, namely $r-k-x_i$ of its neighbors found in $K^i_{r,r}$. Some neighbors are gained, namely $p-x_i$ of the vertices in $S$. Thus, we have that $r = r-k-x_i+p-x_i$. The $r$'s cancel, and we get $0 = p - k - 2x_i$. This is true for every $i \in [k-1]$, so we simply denote the number by $x = x_i$, and get $p = k + 2x$.

    Towards the claim, it remains to show that $x = 0$. Because the neighborhoods of distinct $c_i$ and $c_j$ are disjoint outside $K_{k-1}$, we get that $p \geq (k-1)\cdot x$. We substitute $p$, and get
    $$k+2x \geq (k-1) \cdot x $$ 
    $$k \geq (k-3)\cdot x$$
    $$\frac{k}{k-3} \geq x$$
    Recalling that $k \geq 7$, we have that $x$ is either $1$ or $0$. Assume for the sake of contradiction that $x = 1$. Then without loss of generality, each $c_i$ has some neighbor $a^i_j$ which is in $S$. Since $a^i_j$ had degree $r$ in $G'$, it must hold that $a^i_j$ has equally many neighbors as non-neighbors in $S$. At most one of $a^i_j$'s neighbors is outside of $K^i_{r,r}$, this means that at least $\frac{|S|-3}{2}$ vertices of $K^i_{r,r}$ are in $S$. Because $k \geq 7$ and the $K^i_{r,r}$'s are completely disjoint for different values of $i \in [k-1]$, we get that
    $$|S| \geq \frac{|S|-3}{2} \cdot (k-1) \geq \frac{|S|-3}{2} \cdot 6$$
    $$|S| \geq 3\cdot|S| - 9$$
    $$9 \geq 2\cdot|S|$$

    Seeing that $|S| \geq k-1 \geq 6$, this is a contradiction. Thus, $x$ must be $0$, so $p = k + 2x = k$ and the claim holds.

  We now show that $S \setminus K_{k-1}$ is a clique in $G'$. Assume for the sake of contradiction it is not, and let $u,v \in S \setminus K_{k-1}$ be vertices such that $uv \notin E(G')$. Consider the vertex $u$. By the above claim we know that $u$ does not have a neighbor in $K_{k-1}$. It will thus gain at least $k$ edges going to $K_{k-1} \cup \{v\}$, and lose at most $k-2$ edges going to $S \setminus (K_{k-1} \cup \{u, v\})$. Because $u$ was of degree $r$ in $G'$ yet gained more edges than it lost by the partial complementation, its degree is strictly greater than $r$ in $G \oplus S$. This is a contradiction, hence $S \setminus K_{k-1}$ is a clique in $G'$.

    Because $k \geq 3$, the clique $S \setminus K_{k-1}$ can not be contained in the gadget $\textsc{gdg}_{k, r}$ nor span across both copies of $G$ created in step 2 of the reduction (if that step was applied). It must therefore be contained in the original $G$. Thus, $G$ has a clique of size $k$, and $(G, k)$ is a yes-instance of \KRR{}.
\end{proof}

 Lemmata~\ref{lem_NPhardness1} and ~\ref{lem_NPhardness2} together with Proposition~\ref{rr:prop:krr-hard} conclude the proof of NP-hardness. Membership in NP is trivial, so NP-completeness holds. 
 \end{proof}

\noindent We remark that if $r$ is a constant not given with the input, the problem becomes polynomial time solvable by Theorem~\ref{thm:ddeg}.


 
\section{Conclusion and open problems}\label{secconcl_openpr}
In this paper we initiated the study of \PCG. Many interesting questions remain open. In particular, what is the complexity of the problem when $\mathcal{G}$ is 
%
%
%
%

\begin{itemize}
\item the class of chordal graphs, 
\item the class of interval graphs,
\item the class of graph excluding a path $P_5$ as an induced subgraph,  
\item  the class graphs with max degree $\leq r$, or 
\item the class of graphs with min degree $\geq r$
\end{itemize}

More broadly, it is also interesting to see what happens as we allow more than one partial complementation; how quickly can we recognize the class $\mathcal{G}^{(k)}$ for some class $\mathcal{G}$? It will also be interesting to investigate what happens if we combine partial complementation with other graph modifications, such as the Seidel switch.


\medskip\noindent\textbf{Acknowledgment} We thank Saket Saurabh for helpful discussions, and also a great thanks to the anonymous reviewers who provided valuable feedback.

{{
\bibliography{Complement}
}}




\appendix

\section{Appendix}
\medskip
\noindent

\paragraph*{Clique-width} Let $G$ be a graph and $k$ be a positive integer.
A \emph{$k$-graph} is a graph whose vertices are labeled by
integers from $\{1,2,\dots,k\}$. We call the $k$-graph consisting
of exactly one vertex labeled by some integer from
$\{1,2,\dots,k\}$ an initial $k$-graph. The \emph{clique-width} of $G$, denoted by $\cwd(G)$,
  is the smallest integer $k$ such that $G$ can be
constructed by means of repeated application of the following four
operations on $k$-graphs: {\em $(1)$ introduce}: construction of
an initial $k$-graph labeled by $i$ and denoted by $i(v)$ (that
is, $i(v)$ is a $k$-graph with $v$ as a single vertex and label
$i$), {\em $(2)$ disjoint union} (denoted by $\dot{\cup}$), {\em $(3)$
relabel}: changing all labels $i$ to $j$ (denoted by $\rho_{i\to
j}$), and {\em $(4)$ join}: connecting all vertices labeled by $i$
with all vertices labeled by $j$ by edges (denoted by
$\eta_{i,j}$). Using the symbols of these operations, we can
construct well-formed expressions.  An expression is called
\emph{$k$-expression} for $G$ if the graph produced by performing
these operations, in the order defined by the expression, is
isomorphic to $G$ when labels are removed, and the clique-width of $G$   is the
minimum $k$ such that there is a $k$-expression for $G$.

For integer $k$, we say that a graph class $\mathcal{G}$ is of clique-with at most $k$, if the clique-width of every graph in $\mathcal{G}$ is at most $k$. We also say that a graph class $\mathcal{G}$  is of \emph{bounded clique-width}, if there is a $k$ such that  $\mathcal{G}$ is of clique-with at most $k$.

\paragraph*{Monadic Second Order Logic.} 
\MSOone{} is the sublogic of \MSOtwo{} (Monadic Second Order Logic) without quantifications over edge subsets.  More precisely, 
The syntax of  \MSOone{} of graphs includes the logical connectives $\vee,$ $\land,$ $\neg,$ 
$\Leftrightarrow ,$  $\Rightarrow,$ variables for 
vertices, edges and  sets of vertices, the quantifiers $\forall,$ $\exists$ that can be applied 
to these variables, and the following five binary relations: 
 
\begin{enumerate}

\item 
$u\in U$ where $u$ is a vertex variable 
and $U$ is a vertex set variable; 
\item 
 $\mathbf{inc}(d,u),$ where $d$ is an edge variable,  $u$ is a vertex variable, and the interpretation 
is that the edge $d$ is incident with the vertex $u$; 
\item 
 $\mathbf{adj}(u,v),$ where  $u$ and $v$ are 
vertex variables  and the interpretation is that $u$ and $v$ are adjacent; \item 
 equality of variables representing vertices, edges, sets of vertices, and sets of edges.
\end{enumerate}

 We refer \cite{Courcelle:2012:GSM:2414243} for more information on \MSOone{} and \MSOtwo.
 


\end{document}